\documentclass[10pt]{article}
\usepackage{}

\usepackage{pgf}
\usepackage[square, comma, sort&compress, numbers]{natbib}
\usepackage[english]{babel}
\usepackage[utf8]{inputenc}
\usepackage{datetime}
\usepackage{geometry}
\usepackage{array}
\usepackage{xspace}

\usepackage{amsmath,amsfonts,amssymb,amsopn,amsthm,amscd}
\theoremstyle{plain}
\usepackage{graphicx} 
\usepackage{epstopdf}
\usepackage{enumerate}
\allowdisplaybreaks[4]

\def\Box{\vcenter{\vbox{\hrule\hbox{\vrule
     \vbox to 8.8pt{\hbox to 10pt{}\vfill}\vrule}\hrule}}}

\usepackage{indentfirst}
\newtheorem{thm}{Theorem}[section]
\newtheorem{lem}[thm]{Lemma}

\newtheorem{prop}[thm]{Proposition}

\newtheorem{remark}{Remark}

\begin{document}

\begin{center}

{\large  \bf Minimal Linear Codes Constructed from hierarchical posets with two levels}

\vskip 0.8cm
{\small Xia Wu$^1$, Wei Lu\footnote{Supported by NSFC (Nos. 11971102, 12171241)

  }}$^*$, Xupu Qin$^1$, Xiwang Cao$^2$\\

{\small $^1$School of Mathematics, Southeast University, Nanjing
210096, China}\\
{\small $^2$Department of Math, Nanjing University of Aeronautics and Astronautics, Nanjing 211100, China}\\
{\small E-mail:
wuxia80@seu.edu.cn, luwei1010@seu.edu.cn,  2779398533@qq.com, xwcao@nuaa.edu.cn}\\
{\small $^*$Corresponding author. (Email: luwei1010@seu.edu.cn)}
\vskip 0.8cm
\end{center}

{\bf Abstract:}
J. Y. Hyun, et al. (Des. Codes
Cryptogr., vol. 88, pp. 2475-2492, 2020)  constructed some optimal and minimal binary linear codes generated by one or two order ideals in hierarchical posets of two levels. At the end of their paper, they left an open problem: it also should be interesting to investigate the cases of more than two orders in hierarchical posets with two levels or many levels. In this paper, we use the  geometric method to determine the minimality of linear codes generated by any orders in hierarchical posets with two levels. We generalize their cases of  one or two orders  to any orders and determine the minimality of the linear codes completely.


\section{\bf Introduction}
 Let $q$ be a prime power and $\mathbb{F}_q$ the finite field with $q$ elements. Let $n$ be a positive integer and $\mathbb{F}_q^n$  the vector space with dimension $n$ over $\mathbb{F}_q$. In this paper, all vector spaces are over $\mathbb{F}_q$ and  all vectors are row vectors.  For a vector $\mathbf{v}=(v_1, \dots, v_n)\in\mathbb{F}_q^n$, let Suppt$(\mathbf{v})$ $:= \{1 \leq i\leq n : v_i\neq 0\}$ be the support of $\mathbf{v}$. The \emph{Hamming weight} of vector of $\mathbf{v}$ is wt$(\mathbf{v})$:=\# $\rm{Suppt}(\mathbf{u})$.  For any two vectors $\mathbf{u}, \mathbf{v}\in \mathbb{F}_q^n$, if $\rm{Suppt}(\mathbf{u})\subseteq \rm{Suppt}(\mathbf{v})$, we say that $\mathbf{v}$ covers $\mathbf{u}$ (or $\mathbf{u}$ is covered by $\mathbf{v}$) and write $\mathbf{u}\preceq\mathbf{v}$. Clearly, $a\mathbf{v}\preceq \mathbf{v}$ for all $a\in \mathbb{F}_q$.

 An $[n,k]_q$ linear code $\mathcal{C}$ over $\mathbb{F}_q$ is a $k$-dimensional subspace of $\mathbb{F}_q^n$. Vectors in $\mathcal{C}$ are called codewords. A codeword $\mathbf{c}$ in a linear code $\mathcal{C}$ is called \emph{minimal} if $\mathbf{c}$ covers only the codewords $a\mathbf{c}$ for all $a\in \mathbb{F}_q$, but no other codewords in $\mathcal{C}$. That is to say, if a codeword $\mathbf{c}$  is minimal in  $\mathcal{C}$, then for any codeword $\mathbf{b}$ in $\mathcal{C}$, $\mathbf{b}\preceq \mathbf{c}$ implies that $\mathbf{b}=a\mathbf{c}$ for some $a\in \mathbb{F}_q$.
  For an arbitrary linear code $\mathcal{C}$, it is  hard  to determine the set of its minimal codewords \cite{BMT1978, BN1990}.

If every codeword in  $\mathcal{C}$ is minimal, then $\mathcal{C}$ is said to be a \emph{minimal linear code}. Minimal linear codes have interesting applications in secret sharing \cite{CDY2005, CCP2014, DY2003, M1995, YD2006} and secure two-party computation \cite{ABCH1995, CMP2013}, and could be decoded with a minimum distance decoding method \cite{AB1998}. Searching for minimal linear codes has been an interesting research topic in coding theory and cryptography.

Up to now, there are two approaches to study minimal linear codes. One is  algebraic method and the other  is geometric method. The algebraic method is based on the Hamming weight of the codewords. 
In \cite{AB1998}, Ashikhmin and Barg gave a sufficient condition  on the minimum and maximum nonzero Hamming weights for a linear code to be minimal.
Inspired by Ding's work \cite{D2015, D2016},  many minimal linear codes with $\frac{w_{\rm min}}{w_{\rm max}}>\frac{q-1}{q}$  have been
constructed by selecting the proper defining sets or from functions over finite fields (see \cite{DD2015,HY2016,LCXM2018,SLP2016,SGP2017,TLQZ2016,WDX2015,X2016,YY2017,ZLFH2016}). Cohen et al. \cite{CMP2013} provided an example to show that the condition $\frac{w_{\rm min}}{w_{\rm max}}>\frac{q-1}{q}$  is not necessary for a linear code to be minimal. Ding, Heng and Zhou \cite{DHZ2018, HDZ2018}  generalized this  sufficient condition and derived a sufficient and necessary condition on all Hamming weights for a given linear code to be
minimal.
Based on this sufficient and necessary condition, Ding et al. presented three infinite families of minimal binary linear codes with $\frac{w_{\rm min}}{w_{\rm max}}\leq\frac{1}{2}$ in \cite{DHZ2018} and an infinite family of minimal ternary linear codes with $\frac{w_{\rm min}}{w_{\rm max}}<\frac{2}{3}$ in \cite{HDZ2018}, respectively. In \cite{ZYW2018}, Zhang et al. constructed four families of minimal binary linear codes
with $\frac{w_{\rm min}}{w_{\rm max}}\leq\frac{1}{2}$ from Krawtchouk polynomials. 
In \cite{XQ2019}, Xu and Qu constructed three classes  of minimal linear codes with $\frac{w_{\rm min}}{w_{\rm max}}<\frac{p-1}{p}$ for any odd prime $p$. In \cite{LY2019,MQ2019,XQC2020}, minimal linear codes constructed from functions are introduced.  When using the algebraic method to prove the minimality of a given linear code, one needs to know all the Hamming weights in the code, which is very difficult in general. Even all the
Hamming weights are known, it is also hard to use the algebraic method   to prove the minimality.

Recently, minimal linear codes were characterized by geometric approaches in \cite{ABN2019,LW2019,TQLZ2019}. In \cite{ABN2019,TQLZ2019}, the authors used the cutting blocking sets to study the minimal linear codes. In \cite{LW2019}, the authors used the basis of linear space to study the minimal linear codes. Based on these results, it is  easier to construct minimal linear codes or to prove that some linear codes are minimal, see \cite{ABNR2019,BB2019,BB2021,BCMP2021,MB2019,HN}.


In \cite{HKWY2020}, the authors used the algebraic method to construct some optimal and minimal binary linear codes generated by one or two order ideals in hierarchical posets of two levels. At the end of their paper, they left an open problem: it also should be interesting to investigate the cases of more than two orders in hierarchical posets with two levels or many levels. In this paper, we use the geometric method in \cite{LW2019} to completely determine the minimality of binary linear codes generated by  any  order ideals in hierarchical posets with two levels.

The rest of this paper is organized as follows. In Section \ref{section Preliminaries}, we introduce basic concepts on posets and hierarchical posets with two levels. We also introduce the sufficient and necessary condition in \cite{LW2019}.  In Section \ref{section hierarchical posets}, we completely determine the minimality of  minimal linear codes  constructed from hierarchical posets with two levels. In Section \ref{section Concluding remarks}, we conclude this paper.

\section{\bf Preliminaries}\label{section Preliminaries}
\subsection{Posets}

Let $[n] $ be the set $\{1,2,...,n\}$ and $2^{[n]}$  the power set of  $[n]$. There exists a bijection between  $\mathbb{F}_{2}^{n}$  and  $2^{[n]}$, defined by  $\mathbf{v} \mapsto \operatorname{Suppt}(\mathbf{v}) $. \textbf{Throughout this paper, we will identify a vector in  $\mathbb{F}_{2}^{n}$  with its support}.

We say that  $\mathbb{P}=([n], \preceq)$  is a partially ordered set (abbreviated as a poset) if  $\mathbb{P}$  is a partial order relation on  $[n]$ , that is, for all  $i, j, k \in[n]$  we have that:  $(i)\ i \preceq i $; $(i i)\ i \preceq j$  and  $j \preceq i$  imply  $i=j$ ; $(i i i)\ i \preceq j$  and  $j \preceq k$  imply  $i \preceq k$.

Let  $\mathbb{P}=([n], \preceq)$  be a poset. Two distinct elements  $i$  and $ j$  in $ [n]$  are called comparable if either  $i \preceq j$  or  $j \preceq i$ , and incomparable otherwise.

A nonempty subset  $I $ of  $\mathbb{P}$  is called an order ideal if  $j \in I$  and  $i \preceq j$  imply  $i \in I$. Let $\mathcal{O}_{\mathbb{P}}$ denote the set of all order ideals of $\mathbb{P}.$ For a subset  $E$  of  $\mathbb{P}$, the smallest order ideal of  $\mathbb{P}$  containing  $E$  is denoted by  $\langle E\rangle $. For an order ideal $ I$  of  $\mathbb{P}$, we use  $I(\mathbb{P})$  to denote the set of order ideals of  $\mathbb{P} $  which are contained in  $I$. Let  $\mathcal{I}=\left\{I_{1}, \ldots, I_{m}\right\}$  be a subset of  $\mathcal{O}_{\mathbb{P}}$. We define

$$\mathcal{I}(\mathbb{P})=\left\{J \in \mathcal{O}_{\mathbb{P}}: J \subseteq I \in \mathcal{I}\right\}=\bigcup_{i=1}^{m} I_{i}(\mathbb{P}). $$

\subsection{Hierarchical posets with two levels}
Let  $m$  and  $n$  be positive integers with  $m \leq n.$ We say that  $\mathbb{H}(m, n)=([n], \preceq)$  is a hierarchical poset with two levels if  $[n]$  is the disjoint union of two incomparable subsets  $U=\{1, \ldots, m\}$  and $ V=\{m+1, \ldots, n\} $, and  $i \prec j$  whenever  $i \in U$  and  $j \in V $.

\begin{lem}\label{hmn}\cite[Lemma 4.1]{HKWY2020}
Every order ideal of  $\mathbb{H}(m, n)$  can be expressed by $ A \cup B $ for  $A \subseteq[m],\ B \subseteq   [n] \backslash[m]$, and one of the following holds:
$(i) \ B=\emptyset ;\ (i i)\ B \neq \emptyset$  and $ A=[m]$.
\end{lem}
Let $l:=n-m$. Then $\mathbb{F}_2^n=\mathbb{F}_2^m\oplus\mathbb{F}_2^l$.  Since we  identify a vector in  $\mathbb{F}_{2}^{n}$  with its support in this paper, we use $\alpha=(\beta,\gamma)\in \mathbb{F}_2^n$ to denote the subset of $[n]$ where $\beta\in\mathbb{F}_2^m$ and $\gamma\in\mathbb{F}_2^l$. Then Lemma \ref{hmn} can be described as the following Lemma.
\begin{lem}\label{hmn1}
If $\alpha=(\beta,\gamma)$ is an order ideal of $\mathbb{H}(m,n)$, then one of the following holds:

$(i)$\ $\gamma=\mathbf{0};$

$(ii)$\ $\gamma\neq \mathbf{0}$ and $\beta=\mathbf{1}.$
\end{lem}
Let $I=[m]\cup B$ be an order ideal of  $\mathbb{H}(m,n)$, where $B\neq \emptyset.$
Then $$\emph{I}(\mathbb{P})=\{(\mathbf{1},\gamma)|{\rm{Suppt}}(\gamma)\subseteq B\}\cup\{(\beta, \mathbf{0})|\beta\in \mathbb{F}_2^m\}.$$

\subsection{Minimal linear codes}
All linear codes can be constructed by the following way. Let $k\leq n$ be two positive integers. Let $D:=\{\mathbf{d}_1,...,\mathbf{d}_n\}$  be a multiset, where $\mathbf{d}_1,...,\mathbf{d}_n\in \mathbb{F}_q^k$. Let  Let  $rank(D)$ be the rank of $D$ (it equals to the dimension of the vector space $Span(D)$ over $\mathbb{F}_q$), $$\mathcal{C}=\mathcal{C}(D)=\{{\mathbf{c}\mathbf{(x)}}=\mathbf{c}(\mathbf{x};D)=(\mathbf{xd}_1^{T},...,\mathbf{xd}_n^{T}), \mathbf{x}\in \mathbb{F}_q^k\}.$$
Then $\mathcal{C}(D)$ is an $[n$, rank$(D)]_q$ linear code. We always study the minimality of $\mathcal{C}(D)$ by considering some appropriate  multisets $D$.

A new sufficient and necessary condition is presented as a main result in \cite{LW2019}, for easier reading, we list the results once more. To present the new sufficient and necessary condition in \cite{LW2019}, some concepts are needed.

For any $\mathbf{y}\in \mathbb{F}_q^k$, we define
$$H(\mathbf{y}):=\mathbf{y}^\perp=\{\mathbf{x}\in \mathbb{F}_q^k\mid\mathbf{xy}^{T}=0\},$$
$$H(\mathbf{y},D):=D\cap H(\mathbf{y})=\{\mathbf{x}\in D\mid\mathbf{xy}^{T}=0\},$$
$$V(\mathbf{y},D):={\rm{Span}}(H(\mathbf{y},D)).$$
It is obvious that $H(\mathbf{y},D)\subseteq V(\mathbf{y},D)\subseteq H(\mathbf{y})$.

\begin{prop}\label{21}
For any $\mathbf{x}, \mathbf{y}\in \mathbb{F}_q^k,\ \mathbf{c(x)}\preceq \mathbf{c(y)}$ if and only if $H(\mathbf{y},D)\subseteq H(\mathbf{x},D)$.\end{prop}

Let $\mathbf{y}\in \mathbb{F}_q^k\backslash \{\mathbf{0}\}$. The following lemma   gives a sufficient and necessary condition for the codeword $\mathbf{c(y)}\in \mathcal{C}(D)$ to be minimal. The key idea is that, instead of the weights, we consider the zeros in the codewords: $\mathbf{c(x)}\preceq\mathbf{c(y)}$ if and only if $H(\mathbf{y},D)\subseteq H(\mathbf{x},D)$.

\begin{lem}\label{sn}\cite[Theorem 3.1]{LW2019}
Let $\mathbf{y}\in \mathbb{F}_q^k\backslash \{\mathbf{0}\}$. Then the following three conditions are equivalent:\\
$(1)$ $\mathbf{c(y)}$ is minimal in $\mathcal{C}(D)$;\\
$(2)$ \rm{dim}$V(\mathbf{y},D)=k-1$;\\
$(3)$ $V(\mathbf{y},D)=H(\mathbf{y})$.
\end{lem}

In the following section, we will use the above lemma to consider the minimality of linear codes constructed from hierarchical posets with two levels.

\section{Minimal linear codes constructed from hierarchical posets with two levels}\label{section hierarchical posets}
Let $t$ be a positive integer. Let $\mathbb{P}=\mathbb{H}(m, n)$ and $\mathcal{I}=\left\{I_{1}, \ldots, I_{t}\right\}$ be a subset of $\mathcal{O}_{\mathbb{P}}$, where $I_i=[m]\cup B_i$ are  order ideals of  $\mathbb{H}(m,n)$, where $B_i\neq \emptyset,$ $1\leq i\leq t.$

Define
\begin{equation}
D:=(\mathcal{I}(\mathbb{P}))^c=2^{[n]}\setminus\mathcal{I}(\mathbb{P}),
 \end{equation}
 \begin{equation}
 D_0:=\{(\beta,\gamma)\mid \beta\neq \mathbf{1},\ \gamma\neq \mathbf{0}\},
  \end{equation}
 \begin{equation}
 D_1:=\{(\mathbf{1},\gamma)\mid {\rm{Suppt}}(\gamma)\nsubseteq B_i,\ 1\leqslant i\leqslant t\}.
   \end{equation}
   We can see $D=D_0\cup D_1.$ Since $D_0\subseteq D$, by the following lemma, in order to study the minimality of $\mathcal{C}(D)$, we can first study the minimality of $\mathcal{C}(D_0)$.

\begin{lem}\label{31}\cite[Propsition 4.1]{LW2019}
Let $M\subseteq N$ be two multisets with elements in $\mathbb{F}_q^k$ and {\rm{rank}}$(M)$={\rm{rank}}$(N)=k$. If $\mathcal{C}(M)$ is minimal, then $\mathcal{C}(N)$ is minimal.
\end{lem}
In the rest of this paper, let $\mathbf{e}_i$
denote the $i^{th}$ standard basis vector.

First we consider the minimality of $\mathcal{C}(D_0)$.
\begin{thm}\label{32}
Let $n=m+l,$ where $m,\ n,\ l\in \mathbb{N}^+$. \\
$(1)$ When $m>2$ and $l\geq 2$, $\mathcal{C}(D_0)$ is minimal.\\
$(2)$ When $m=2$ and $l\geq 2$, $\mathcal{C}(D_0)$ is not minimal.\\
$(3)$ When $m\geq2$ and $l=1$, $\mathcal{C}(D_0)$ is not minimal.\\
$(4)$  When $m=1$ and $l\geq 1$, $\mathcal{C}(D_0)$ is minimal.
\end{thm}

\begin{proof}
(1) When $m>2$ and $l\geq 2$, rank$(\mathcal{C}(D_0))=n$. For any $\mathbf{0}\neq \mathbf{u}=(\mathbf{u_1},\mathbf{u_2})\in \mathbb{F}_2^m\oplus\mathbb{F}_2^l$, we will prove that all the codewords $\mathbf{c}(\mathbf{u})\in \mathcal{C}(D_0)$ are minimal in three different cases.

{\bf Case 1:} When $\mathbf{u_1}=\mathbf{0}$ and $\mathbf{u_2}\neq\mathbf{0}$, there exists $\{\gamma_1$,..., $\gamma_{l-1}\}$ which is linearly
independent and satisfying $\mathbf{u_2}\cdot \gamma_i=0,$ $i=1,...,l-1$. It is easy to verify that
\begin{center}
 $(\mathbf{e_1},\gamma_1),\ ...,(\mathbf{e_m},\gamma_1),(\mathbf{0},\gamma_1),...,(\mathbf{0},\gamma_{l-1})\in H(\mathbf{u},D_0)$,
\end{center}
  and this vector group is linearly independent. By Lemma \ref{sn}, $\mathbf{c}(\mathbf{u})$ is minimal.

{\bf Case 2:}  When $\mathbf{u_1}\neq\mathbf{0}$ and $\mathbf{u_2}=\mathbf{0}$, there exists $\{\beta_1,...,\beta_{m-1}\}$ which is linearly independent, and satisfying $\mathbf{u_1}\cdot \beta_i=0$. If for some $i$, $\beta_i=\mathbf{1}$, without loss of generality, we set $\beta_{m-1}=\mathbf{1}$. In this case,  we replace $\beta_{m-1}$ with $\beta_{1}+\beta_{m-1}$. Since $m>2$, we can see that the vector group $\{\beta_1,...,\beta_{m-2},\beta_{m-1}\}$ is linearly independent, $\mathbf{u_1}\cdot \beta_i=0$ and $\beta_i\neq \mathbf{1},\ 1\leq i\leq m-1$.  It is easy to verify that $(\mathbf{0},\mathbf{e_1}),\ ...,(\mathbf{0},\mathbf{e_l}),(\beta_1,\mathbf{e_1}),...,(\beta_{m-1},\mathbf{e_1})\in H(\mathbf{u,}\ D_0)$, and  the vector group is linearly independent. By Lemma \ref{sn}, $\mathbf{c}(\mathbf{u})$ is minimal.

{\bf Case 3:} When $\mathbf{u_1}\neq\mathbf{0}$ and $\mathbf{u_2}\neq \mathbf{0}$, since $\mathbf{u_2}\neq\mathbf{0}$, there exist $\gamma_1,...,\gamma_{l-1}$, such that $\mathbf{u_2}\cdot \gamma_i={0}=-\mathbf{u_1}\cdot\mathbf{0},$ $i=1,...,l-1$, and there exists $\gamma_0$ such that $-\mathbf{u_2}\cdot \gamma_0\neq 0.$ As the discussion in Case 2, we know that there exists $\beta_1,...,\beta_{m}$ which is linearly independent, such that $\mathbf{u_1}\cdot \beta_i=-\mathbf{u_2}\cdot \gamma_0$ and $\beta_i\neq \mathbf{1}$. It is easy to verify that $(\mathbf{0},\gamma_1),\ ...,(\mathbf{0},\gamma_2),(\beta_1,\gamma_0),...,(\beta_{m},\gamma_0)\in H(\mathbf{u,}\ D_0)$, and  the vector group is linearly independent. By Lemma \ref{sn}, $\mathbf{c}(\mathbf{u})$ is minimal.

(2) When $m=2$ and $l\geq 2$, rank$(\mathcal{C}(D_0))=n$, let $\mathbf{u}=(1,1,\mathbf{0}) $ and $\mathbf{v}=(1,0,\mathbf{0})$. First we prove that $H(\mathbf{u},D_0)\subseteq H(\mathbf{v},D_0)$. For any $\alpha=(\beta,\gamma)\in H(\mathbf{u},D_0)$, $ \mathbf{u}\cdot \alpha=(1,1)\cdot \beta+\mathbf{0}\cdot \gamma=0$, we get $\beta=(0,0)$ or $(1,1)$. Since $\alpha=(\beta,\gamma)\in D_0$, we get $\beta=(0,0)$ and $\alpha\cdot\mathbf{v}=0$. Then $\alpha\in H(\mathbf{v},D_0)$ and $H(\mathbf{u},D_0)\subseteq H(\mathbf{v},D_0)$. Then by Proposition \ref{21}, we get $\mathbf{c(v)}\preceq \mathbf{c(u)}$. For $\alpha=(0,1,\mathbf{e_1})\in D_0$, $\mathbf{u}\cdot \alpha=1$, $\mathbf{v\cdot \alpha=0}$, so  $\mathbf{c(u)}\neq \mathbf{c(v)}$. It is easy to verify that $\mathbf{c(u)}\neq\mathbf{0}$ and  $\mathbf{c(v)}\neq \mathbf{0}$, thus $\mathbf{c(u)}$ is not a minimal codeword and  $\mathcal{C}(D_0)$ is not minimal.

(3) When $m\geq 2$ and $l=1$, rank$(\mathcal{C}(D_0))=n$. Let $\mathbf{u}=(\mathbf{0},1) $ and $\mathbf{v}=(\mathbf{e_1}+\mathbf{e_2},1)$. We prove that $H(\mathbf{u},D_0)\subseteq H(\mathbf{v},D_0)$. For any $\alpha=(\beta,\gamma)\in H(\mathbf{u},D_0)$, we know that $\mathbf{u}\cdot \alpha=\mathbf{0}\cdot \beta+{1}\cdot \gamma=0$, then $\gamma=0$,  it  conflicts with the fact that $\alpha\in D_0$. So $H(\mathbf{u},D_0)=\emptyset$ and $H(\mathbf{u},D_0)\subseteq H(\mathbf{v},D_0)$.  By Proposition \ref{21}, we get $\mathbf{c(v)}\preceq \mathbf{c(u)}$. For $\alpha=(\mathbf{e_1},1)\in D_0$, $\mathbf{u}\cdot \alpha=1$, $\mathbf{v\cdot \alpha=0}$, so  $\mathbf{c(u)}\neq \mathbf{c(v)}$. It is easy to verify that $\mathbf{c(u)}\neq\mathbf{0}$ and  $\mathbf{c(v)}\neq \mathbf{0}$, thus $\mathbf{c(u)}$ is not a minimal codeword and  $\mathcal{C}(D_0)$ is not minimal.

(4) When $m=1$ and $l\geq 1$, $D_0=\{(\beta,\gamma)\mid \beta= 0,\ \gamma\neq \mathbf{0}\}$, we can see rank$(\mathcal{C}(D_0))=n-1$. For any $\mathbf{u}=(u_1, \mathbf{u_2})\in \mathbb{F}_2^n$,
$$\mathbf{c(u)}=(u_1\cdot 0+\mathbf{u_2}\cdot \gamma)_{\gamma\neq \mathbf{0}}=(\mathbf{u_2}\cdot \gamma)_{\gamma\neq \mathbf{0}}\in \mathcal{C}(\mathbb{F}_2^l\setminus\mathbf\{{0}\}).$$
It is easy to see $\mathcal{C}(D_0)=\mathcal{C}(\mathbb{F}_2^l\setminus\mathbf\{{0}\})$ and then  $\mathcal{C}(D_0)$ is minimal.

This completes the proof.
\end{proof}


Now we consider the minimality of $\mathcal{C}(D)$.
\begin{thm}\label{33}
Let $n=m+l,$ where $m,\ n,\ l\in \mathbb{N}^+$. \\
$(1)$ When $m>2$ and $l\geq 2$, $\mathcal{C}(D)$ is minimal.\\
$(2)$ When $m=2$ and $l\geq 2$, $\mathcal{C}(D)$ is  minimal if and only if for any $B_i$, $|B_i|<l,\ 1\leq i\leq t$.\\
$(3) $ When $m\geq2$ and $l=1$, $\mathcal{C}(D)$ is not minimal.\\
$(4)$  When $m=1$ and $l\geq 1$, there are three cases in the following:\\
\ \ $(i)$ If there exists $B_{i_0}$ such that $|B_{i_{0}}|=l$, $1\leq i_0\leq t$, then $\mathcal{C}(D)$ is minimal;\\
\ \ $(ii) $ If for all $B_i$, $|B_i|\leq l-1$, $1\leq i\leq t$,  and there exists $B_{i_0}$ such that $|B_{i_{0}}|=l-1,$ $1\leq i_0\leq t$, then $\mathcal{C}(D)$ is not minimal;\\
\ \ $(iii)$ If for all $B_i$, $|B_i|< l-1$, $1\leq i_0\leq t$, then $\mathcal{C}(D)$ is minimal.
\end{thm}

\begin{proof}
$(1)$ When $m>2$ and $l\geq 2$, by Theorem \ref{32} and Lemma \ref{31}, the result is immediate.

$(2)$ $"\Leftarrow"$ For any $\mathbf{0}\neq \mathbf{u}=(\mathbf{u_1},\mathbf{u_2})\in \mathbb{F}_2^2\oplus\mathbb{F}_2^l$, we will prove the codeword $\mathbf{c}(\mathbf{u})\in \mathcal{C}(D)$ is minimal in three different cases.

{\bf Case 1:} When $\mathbf{u_1}=\mathbf{0}$ and $\mathbf{u_2}\neq\mathbf{0}$, the proof is the same as the Case 1 in Theorem \ref{32}$(1)$.

{\bf Case 2:}  When $\mathbf{u_1}\neq\mathbf{0}$ and $\mathbf{u_2}=\mathbf{0}$:

 if $\mathbf{u_1}=(1,1)$, set $\beta=(1,1)$,  it is easy to verify that $(\mathbf{0},\mathbf{e_1}),\ ...,(\mathbf{0},\mathbf{e_l}),(\beta,\mathbf{1})\in H(\mathbf{u},D)$, and  the vector group is linearly independent;

  if $\mathbf{u_1}=(1,0)$ or $(0,1)$, we set $\beta=(0,1)$ or $(1,0)$ respectively, we can verify that $(\mathbf{0},\mathbf{e_1}),\ ...,(\mathbf{0},\mathbf{e_l}),(\beta,\mathbf{e_1})\in H(\mathbf{u},D_0)$, and  the vector group is linearly independent. By Lemma \ref{sn}, $\mathbf{c}(\mathbf{u})$ is minimal.

{\bf Case 3:} When $\mathbf{u_1}\neq\mathbf{0}$ and $\mathbf{u_2}\neq \mathbf{0}$, since $\mathbf{u_2}\neq\mathbf{0}$, there exists $\gamma_1,...,\gamma_{l-1}$, such that $\mathbf{u_2}\cdot \gamma_i=0=-\mathbf{u_1}\cdot\mathbf{0},$ $i=1,...,l-1$.

 If $\mathbf{u_1}=(1,1)$, set $\beta_1=(1,0)$ and $\beta_2=(0,1)$. Since $\mathbf{u_2}\neq \mathbf{0}$, there exists $\gamma_0\neq \mathbf{0}$, such that $-\mathbf{u_2}\cdot \gamma_{0}=1= \mathbf{u_1}\cdot \beta_1=\mathbf{u_1}\cdot \beta_2$.
 We can verify that $(\mathbf{0},\gamma_1),\ ...,(\mathbf{0},\gamma_{l-1}),(\beta_1,\gamma_0),(\beta_2,\gamma_0)\in H(\mathbf{u},D_0)$, and  the vector group is linearly independent;

 If $\mathbf{u_1}=(1,0)$ or $(0,1)$, we set $\beta_1=(1,0)$ and $\beta_2=(0,1)$. Since $\mathbf{u_2}\neq \mathbf{0}$, there exists $\gamma_{01}\neq \mathbf{0}$ and $\gamma_{02}\neq \mathbf{0}$, such that $-\mathbf{u_2}\cdot \gamma_{01}=- \mathbf{u_1}\cdot \beta_1$ and $-\mathbf{u_2}\cdot \gamma_{02}=- \mathbf{u_1}\cdot \beta_2$.
 We can verify that $(\mathbf{0},\gamma_1),\ ...,(\mathbf{0},\gamma_{l-1}),(\beta_1,\gamma_{01}),(\beta_2,\gamma_{02})\in H(\mathbf{u},D_0)$, and  the vector group is linearly independent. By Lemma \ref{sn}, $\mathbf{c}(\mathbf{u})$ is minimal.

 $"\Rightarrow"$  If there exists $B_{i_0}$ such that $|B_{i_{0}}|=l$, then $D_1=\emptyset$ and $D=D_0$, by Theorem \ref{32}(2), $\mathcal{C}(D)=\mathcal{C}(D_0)$ is not minimal. The result follows.

 $(3)$ When $m\geq2$ and $l=1$, $D_1=\{(\mathbf{1},\gamma)\mid {\rm{supp}}(\gamma)\nsubseteq B_i,\ 1\leqslant i\leqslant t\}=\emptyset$, so $D=D_0$. By Theorem \ref{32}(3), $\mathcal{C}(D)=\mathcal{C}(D_0)$ is not minimal. The result follows.

 $(4)(i)$  When $m=1$ and $l\geq 1$,
if there exists $B_{i_0}$ such that $|B_{i_{0}}|=l$, $1\leq i_0\leq t$, then $D_1=\emptyset$. By Theorem \ref{32}(4), $\mathcal{C}(D)=\mathcal{C}(D_0)$ is minimal. The result follows.

$(4)(ii)$ Let $\mathbf{u}=(0,\mathbf{1}-B_{i_0})$. Since $|B_{i_0}|=l-1,$ there exists $\mathbf{e}_{j_0}$, such that $\mathbf{u}=(0,\mathbf{e}_{j_0})$. Let $\mathbf{v}=(1,\mathbf{e}_{j_0})$. For any $\alpha=(\beta,\gamma)\in H(\mathbf{u},D)$, $\mathbf{u}\cdot \alpha=0\cdot \beta+\mathbf{e}_{j_0}\cdot \gamma=0,$ we know that supp$(\gamma)\subseteq B_{i_0}$, and $\alpha \notin D_1$. Since $\alpha \in D$, we get $\alpha \in D_0$, then $\beta=0.$ So $\mathbf{v}\cdot \alpha=1\cdot\beta+\mathbf{e}_{j_0}\cdot \gamma=0$, $\alpha\in H(\mathbf{v},D)$. We get
$H(\mathbf{u},D)\subseteq H(\mathbf{v},D)$, by Proposition \ref{21},  $\mathbf{c(v)}\preceq \mathbf{c(u)}$. We set $\alpha_1=(0,\mathbf{e}_{j_0})\in D$, then $\mathbf{u}\cdot \alpha_1=1$, $\mathbf{v}\cdot \alpha_1=1$, so  $\mathbf{c(u)}\neq \mathbf{0}$ and $ \mathbf{c(v)}\neq \mathbf{0}$. For $\alpha_2=(1,\mathbf{1})\in D$, $\mathbf{u}\cdot \alpha_2=1$, $\mathbf{v}\cdot \alpha_2=0$,  $\mathbf{c(u)}\neq \mathbf{c(v)}$, thus $\mathbf{c(u)}$ is not a minimal codeword and  $\mathcal{C}(D)$ is not minimal.

$(4)(iii)$ For any $\mathbf{0}\neq \mathbf{u}=({u_1},\mathbf{u_2})\in \mathbb{F}_2\oplus\mathbb{F}_2^l$, we will prove the codeword $\mathbf{c}(\mathbf{u})\in \mathcal{C}(D_0)$ is minimal in three different cases.

{\bf Case 1:} When ${u_1}={0}$ and $\mathbf{u_2}\neq\mathbf{0}$, there exists $\gamma_1$, $\gamma_2$,...,$\gamma_{l-1}$ which is linearly independent, such that $\mathbf{u_2}\cdot {\gamma_i}=0$. (1) If wt$(\mathbf{u_2})$ is even, we set $\gamma_0=\mathbf{1}$; (2) if wt$(\mathbf{u_2})$ is odd, since $\mathbf{u_2}\neq\mathbf{0}$,  there exists $\gamma_0$ such that wt$(\gamma_0)=l-1$ and $\mathbf{u_2}\cdot {\gamma_0}=0$. It is easy to verify that $({0},\gamma_1),\ ...,({0},\gamma_{l-1}),(1,\gamma_0)\in H(\mathbf{u},D)$, and  the vector group is linearly independent. By Lemma \ref{sn}, $\mathbf{c}(\mathbf{u})$ is minimal.

{\bf Case 2:}  When ${u_1}\neq{0}$ and $\mathbf{u_2}=\mathbf{0}$, it is easy to verify that $({0},\mathbf{e}_1),\ ({0},\mathbf{e}_2),...,({0},\mathbf{e}_l)\in H(\mathbf{u},D_0)$, and  the vector group is linearly independent. By Lemma \ref{sn}, $\mathbf{c}(\mathbf{u})$ is minimal.

{\bf Case 3:} When ${u_1}\neq{0}$ and $\mathbf{u_2}\neq \mathbf{0}$, since $\mathbf{u_2}\neq\mathbf{0}$, there exists $\gamma_1,...,\gamma_{l-1}$ which is linearly independent, such that $\mathbf{u_2}\cdot \gamma_i=0$, $i=1,...,l-1$.
 (1) If wt$(\mathbf{u_2})$ is odd, we set $\gamma_0={1}$; (2) if wt$(\mathbf{u_2})$ is even, since $|B_i|<l-1$,  there exists $\gamma_0$ such that wt$(\gamma_0)=l-1$ and $\mathbf{u_2}\cdot {\gamma_0}=1$. It is easy to verify that $({0},\gamma_1),\ ...,({0},\gamma_{l-1}),(1,\gamma_0)\in H(\mathbf{u},D)$, and  the vector group is linearly independent. By Lemma \ref{sn}, $\mathbf{c}(\mathbf{u})$ is minimal.

 This completes the proof.

\end{proof}

\begin{remark}
Here we compare the results in Theorem \ref{33} with those in \cite{HKWY2020}. When the number of order ideals  $t=1$, \cite[Theorem 6.1(3)]{HKWY2020} determine the minimality in all cases. When $t=2$, only in two cases, the minimality is determined:(1) $|B_1|=|B_2|=1$, $1<m\leq n-2$, see \cite[Theorem 6.2]{HKWY2020}; (2) $B_1\bigcap B_2=\emptyset $, $|B_1|+|B_2|=n-m$,
and max$\{|B_1|,\ |B_2|\}\leq n-2$, see  \cite[Theorem 6.4]{HKWY2020}. While in Theorem \ref{33} in this paper, for all positive integer $t$, we determine the minimality in all cases.
\end{remark}

\section {\bf{Concluding remarks}}\label{section Concluding remarks}
In \cite{HKWY2020}, J. Y. Hyun, H. K. Kim, et al.  constructed some optimal and minimal binary linear codes generated by one or two order ideals in hierarchical posets of two levels, and used the algebraic method to determine the minimality of the linear codes in some cases. In this paper, we use the geometric method   in \cite{LW2019} to  determine the minimality of binary linear codes generated by any   order ideals in hierarchical posets with two levels in all cases.  Our results have  answered the open problem in the end of \cite{HKWY2020}. We generalize the cases with one or two orders in \cite{HKWY2020} to any orders   and determine the minimality of the linear codes completely. It is hardly to use the algebraic method in \cite{HKWY2020} to prove the results in this paper, while it is very easy to use the geometric method to prove them.  It is also interesting to consider the cases where hierachical posets with more than two levels.

 {}
\end{document}